\documentclass{lmcs}
\pdfoutput=1
\usepackage[utf8]{inputenc}

\usepackage{lastpage}
\lmcsdoi{21}{4}{23}
\lmcsheading{}{\pageref{LastPage}}{}{}%
{Apr.~07,~2025}{Nov.~21,~2025}{}

\usepackage{amsmath}
\usepackage{amsthm}
\usepackage{amssymb}

\newcommand{\comp}{\mathrm{comp}}
\renewcommand{\P}{\mathsf{P}}
\newcommand{\PSPACE}{\mathsf{PSPACE}}
\newcommand{\EXP}{\mathsf{EXP}}
\newcommand{\EXPax}{\mathrm{Exp}}
\newcommand{\acc}{\mathrm{acc_{\EXP}}}
\newcommand{\NEXP}{\mathsf{NEXP}}
\newcommand{\coNEXP}{\mathsf{coNEXP}}
\newcommand{\poly}{\mathsf{poly}}
\newcommand{\PHP}{\mathrm{PHP}}

\newcommand{\NN}{\mathbb{N}}
\newcommand{\St}{\tilde \Sigma^{1,b}}
\newcommand{\Log}{\mathrm{Log}}

\let\phi\varphi
\let\epsilon\varepsilon

\begin{document}

\title[On the consistency of stronger lower bounds for NEXP]
{On the consistency of stronger \texorpdfstring{\\}{} lower bounds for NEXP }

\thanks{Supported by the Czech Academy of Sciences (RVO 67985840) and GA\v{C}R grant 23-04825S}

\author{Neil Thapen}
\address{Institute of Mathematics, Czech Academy of Sciences}
\email{thapen@math.cas.cz}  

\begin{abstract}
It was recently shown by Atserias, Buss and M\"uller
that the standard comple\-xi\-ty-theoretic conjecture $\NEXP \not \subseteq \P / \poly$
is consistent with the relatively strong bounded arithmetic theory $V^0_2$,
which can prove a substantial part of
complexity theory. We observe that their approach
can be extended to show that the stronger conjectures 
$\NEXP \not \subseteq \EXP / \poly$
and $\NEXP  \not \subseteq \coNEXP$
are consistent with a stronger theory,
which includes every true universal number-sort sentence.
\end{abstract}

\maketitle

\section{Introduction}

The bounded arithmetic hierarchy
$S^1_2 \subseteq T^1_2 \subseteq S^2_2 \subseteq \dots$,
with union~$T_2$,
is a well-studied family of first-order theories that plausibly captures
the kind of reasoning one can do if one is limited to using concepts in
the polynomial hierarchy~\cite{buss1985bounded}.
It has long been of interest how much of mathematics can be
carried out in this setting~\cite{pww:primes}.
Many results in complexity theory can be formalized in it, at low levels in the hierarchy;
some more recent examples are the PCP theorem, 
Toda's theorem, the Schwartz-Zippel lemma
and many circuit lower bounds~\cite{
pich2015logical, buss2015collapsing, atserias2024feasibly, muller2020feasibly}.

Now consider a theory~$T$,
such as $T_2$ or some fragment of it, which is known to formalize
a substantial part of complexity theory, and take a
conjecture $\mathcal{C}$ that you would like to prove. If you can show that $\mathcal{C}$
is unprovable in~$T$, then one interpretation of this is that the methods
 which work for large parts of complexity theory are not enough 
to prove~$\mathcal{C}$, and you need to try something new;
see~\cite{pich2021strong} for some work in this direction.
On the other hand, if the negation $\neg \mathcal{C}$ is unprovable in~$T$,
this shows that $\mathcal{C}$ is at least consistent with~$T$ and thus
with a large part of complexity theory. Concretely, there is a well-behaved structure
(a~model of~$T$)
which satisfies many of the complexity-theoretic properties of the real world
and in which~$\mathcal{C}$ is true. Interpreted optimistically, this is a partial result
in the direction of showing~$\mathcal{C}$ is true in the real world~\cite{krajicek1995bounded}.

The $T^i_2$ hierarchy is not able to reason naturally about
complexity phenomena at the level of $\PSPACE$, $\EXP$ or above, since 
by design it is limited to working with polynomial-length strings.
Already~\cite{buss1985bounded} introduced a 
stronger hierarchy $V^0_2 \subseteq V^1_2 \subseteq \dots$
of two-sorted or ``second-order'' theories which also work with larger
objects which we will call here sets, but which could just as
well be called exponential-length strings. The base theory
$V^0_2$ in the hierarchy is a conservative extension of $T_2$, and the next level~$V^1_2$
can prove many basic properties of~$\EXP$, such as that every exponential-time
machine (even with an oracle) has a computation.

It was recently shown in~\cite{atserias2024consistency} that the
standard  complexity-theoretic conjecture $\NEXP \not \subseteq \P / \poly$
is consistent with~$V^0_2$. That is,

\begin{thmC}[\cite{atserias2024consistency}] \label{the:abm}
$V^0_2 \nvdash \NEXP \subseteq \P / \poly$.
\end{thmC}

The authors suggest this is the best currently available evidence for the truth of
the conjecture.
The purpose of this  note is to show that their proof of Theorem~\ref{the:abm},
based on the well-known unprovability of the pigeonhole principle
in (relativized) $T_2$ or $V^0_2$,
can be extended to show that the  stronger
conjecture $\NEXP \not \subseteq \EXP / \poly$ is consistent with a  stronger theory:

\begin{thm} \label{the:exp_poly}
$V^0_2 + \forall \St_1(\NN) \nvdash \NEXP  \subseteq \EXP / \poly$.
\end{thm}

See Definition~\ref{def:St_1} below for the precise definition of $\forall \St_1(\NN)$, 
but it is roughly all sentences $\forall \vec x \phi(\vec x)$ which are true in $\NN$,
where $\vec x$ are first-order and $\phi(\vec x)$ is a $\NEXP$ property.
It thus contains
every true $\Pi_1$ number-sort sentence, including 
a fortiori all $\Pi_1$ number-sort
consequences of $V^1_2$, and also contains the axiom 
that every exponential-time machine has a computation on every input.

Furthermore by a slightly different argument we show (with a caveat 
that the result depends on a choice to formalize $\coNEXP$ without parameters --
see the discussion at the end of the paper):

\begin{thm} \label{the:conexp}
$V^0_2 + \forall \St_1(\NN) \nvdash \NEXP  \subseteq \coNEXP$.
\end{thm}

These two theorems could be taken as evidence towards the conjectures
$\NEXP \not \subseteq \EXP / \poly$  and
$\NEXP \not \subseteq \coNEXP$  being true;
on the other hand, the theorems are themselves not so difficult to prove
from the pigeonhole principle lower bound,
as was already observed in~\cite{atserias2024consistency},
and it is tempting to conclude that this work rather
shows that~$V^0_2$, even when significantly strengthened, 
is not really equipped to reason nontrivially
about set-sort quantification.

This work also suggests a slightly different perspective on one of the observations
of~\cite{atserias2024consistency}, that $\NEXP \not \subseteq \P / \poly$
is consistent with the considerable amount of complexity theory that
can be formalized in $T_2$. By our construction, the  argument of~\cite{atserias2024consistency} 
can be made to show that the conjecture is consistent with \emph{all} true
statements of complexity theory that can be written as~$\Pi_1$
number-sort statements, since these are included for free in $\forall \St_1(\NN)$.

Does this mean that the prior work on formalizing complexity theory inside~$T_2$
is irrelevant for the implications of these consistency results?
I would say no, since to talk about $\NEXP$ we are working
with set-sort objects, as this is how we choose to formalize exponential-length computations.
The work on formalizing complexity still holds in the presence of such objects,
in that (as far as I know) it all relativizes and is still valid for reasoning about
machines which have access to these objects as oracles,
while our theory $\forall \St_1(\NN)$ has nothing
to say about such a situation.

We have not been able to combine Theorems~\ref{the:exp_poly}
and~\ref{the:conexp} into the natural next step, the consistency
of $\NEXP \not \subseteq \coNEXP / \poly$ with our theory.
This would be particularly interesting as $\NEXP \subseteq \coNEXP / \poly$
is in fact true, as can be shown by a  census argument reported
in~\cite{buhrman2009unconditional}
(although it should be kept in mind that all our constructions rely on a false
statement about $\NEXP$, namely the existence of a function violating the pigeonhole principle,
being consistent with our theory). We comment on one of the
difficulties here at the end of the paper.

Our proofs are self-contained and do not rely on~\cite{atserias2024consistency},
but we assume some  knowledge of bounded arithmetic.
We give an overview below to fix notation, generally following~\cite{atserias2024consistency}.
For  details see e.g.~\cite{buss1985bounded, krajicek1995bounded, buss1998first}.

\medskip

Bounded arithmetic in the style of~\cite{buss1985bounded} has variables $x,y, \dots$ 
ranging over numbers, which we identify when convenient with binary strings.
It has the first-order language $x \le y$, $0$, $1$, $x+y$, $x \cdot y$, $\lfloor x/2 \rfloor$,
$x \# y$, $|x|$ and built-in equality~$x=y$. Here $|x|$ is the length of~$x$ in binary,
and the smash function $\#$ is a weak form of exponentiation defined so that 
$|x \# y| = |x|\cdot|y|$, so the main effect of the totality of smash is that
lengths are closed under polynomials. We will use the abbreviation $x {\in} \Log$
for $\exists y (x = |y|)$ and will write expressions like $2^{x}$ for such numbers,
so as in~\cite{atserias2024consistency} a formula of the form e.g.
$\forall x {\in} \Log ( \dots 2^{x^2} \dots)$ stands for 
$\forall x \forall y ( x = |y| \rightarrow \dots y\#y \dots )$.

A bounded formula is one in which every quantifier is bounded, of the 
form $\forall x {<}t$ or $\exists y {<} t$ for some term~$t$.
We write $\Sigma^b_\infty$ for the set of all bounded formulas.
By $\Pi_1$ above we mean the set of universal closures of $\Sigma^b_\infty$
formulas.

 The theory $T_2$~\cite{buss1985bounded} consists of a set of axioms BASIC,
which are  bounded formulas fixing the basic properties of the 
language, and the induction scheme $\Sigma^b_\infty$-IND, that is, the axiom
\[
 \phi(0) \wedge \forall y{<}z [ \phi(y) \rightarrow \phi(y+1) ] \, \rightarrow \, \phi(z)
\]
for every $\Sigma^b_\infty$ formula $\phi$, which may include other parameters.
The theory $T_2$ is essentially the same as $I\Delta_0 + \Omega_1$~\cite{pww:primes}.

In \emph{two-sorted} bounded arithmetic we add to the language new variables
$X, Y, \dots$ of the \emph{set sort}, representing bounded sets, 
and a relation $x \in X$ between the number
and set sorts. For $\vec X$ a tuple of set variables, we write $\Sigma^b_\infty(\vec X)$
for the set of bounded formulas which may now include these variables, and 
$T_2(\vec X)$ for BASIC$ + \Sigma^b_\infty(\vec X)$-IND
(this is technically slightly different from the way similar notation
is used in~\cite{atserias2024consistency}).
Here we do not allow any quantification over  set-sort variables, so these really
play the role of undefined ``oracle'' predicate symbols.
A theory of this form is sometimes called \emph{relativized}.

Note that, for simplicity, we do not formally enforce that 
these sets are bounded.  For example we do not introduce
a number-valued symbol $|X|$ for the maximum element 
of~$X$, or any axiom like $\forall x ( x \in X \rightarrow x \le |X|)$. 
The bounds on the number-sort in $\Sigma^b_\infty(\vec X)$ formulas 
guarantee that they only access a bounded part of the sets $\vec X$,
and it does not matter what these sets contain beyond this bound.

We write set-sort quantification as e.g. $\exists X$ rather than 
$\exists_2 X$. 
We will rely on capitalization to distinguish the sorts of variables.
A $\Sigma^{1,b}_0$ formula is a formula in the two-sorted language which
can freely use set-sort variables, which contains no set-sort quantifiers,
and in which every number-sort quantifier is bounded.
A~$\Sigma^{1,b}_1$~formula can further contain set-sort 
existential quantifiers (but not inside negations).
 The theory $V^0_2$ extends $T_2$ by 
  the bounded comprehension and induction schemes for $\Sigma^{1,b}_0$~formulas.
 The theory $V^1_2$ further adds induction for $\Sigma^{1,b}_1$~formulas.
  
We introduce the ad-hoc notation $\tilde \Sigma^{1,b}_1$
for formulas of the form $\exists \vec X \phi(\vec X, \vec z)$
where $\phi$ is $\Sigma^{1,b}_0$ and contains no set-sort variables
other than $\vec X$, that is, where~$\phi$ is~$\Sigma^b_\infty(\vec X)$.
These express precisely the $\NEXP$ predicates.
 
 \begin{defi} \label{def:St_1}
The  theory $\forall \St_1(\NN)$ consists of
every \emph{true} sentence of the form $\forall \vec z \exists \vec X \phi(\vec X, \vec z)$,
where $\phi$ is $\Sigma^b_\infty(\vec X)$.
\end{defi}

Here ``true'' means true in the standard model, that is, in the two-sorted structure where
the number-sort is $\NN$ and the set-sort is all bounded subsets of~$\NN$.
Note that these sentences do not contain any universal set-sort quantifiers.

\section{$\NEXP$ and $\EXP / \poly$}

Let $\comp(e,t,W)$ be a $\Sigma^b_\infty(W)$ formula
expressing that $W$ is a computation of the universal Turing machine
on input $e$ running for time $t$ in space $t$. We may think of the first
$t^2$ entries of $W$ as
encoding of a $t \times t$ grid,
where the $i$th column describes the contents of the tape at time $i$;
we also use part of~$W$ to encode the sequence of states and head positions.

We take $\EXPax$ to be the axiom $\forall e \forall t \exists W \comp(e,t,W)$
expressing that every exponential-time machine has a computation\footnote{
Our axiom $\EXPax$ should not be confused
with the complexity class $\EXP$, or with the standard notation exp 
for the axiom asserting that
exponentiation is total on the number sort.
}. This captures exponential-time because we think of a number $x$
as taking an exponential \emph{value} compared to its binary \emph{length} $|x|$.
For example, to express instead that every polynomial-time machine has a computation
we would replace $\comp(e,t,W)$ with~$\comp(e,|t|,W)$; furthermore in this case instead of $W$ we
could use a number-sort variable~$w$ bounded by roughly~$t \# t$ to encode the computation.

The formula $\EXPax$ is $\forall \St_1$, without any set-sort parameters
or universal set-sort quantifiers. It is true in the standard model,
and so is included in $\forall \St_1(\NN)$. 
If we allowed set-sort parameters as oracle inputs to the machine,
 $\EXPax$ would become a much stronger axiom. It would imply, for example, that 
every exponential-size circuit has a computation, and in fact would be enough to prove
the pigeonhole principle, destroying our main construction below.

We start by proving a slightly simpler version of Theorem~\ref{the:exp_poly}:

\begin{prop} \label{pro:exp_poly_simple}
$V^0_2 + \EXPax \nvdash \NEXP \subseteq \EXP / \poly$.
\end{prop}

To formalize $\EXP / \poly$ we introduce
a formula $\acc(e, t)$ to express that the universal deterministic Turing machine,
run for time and space $t$ on input~$e$, accepts. Note that over
$V^0_2 + \EXPax$ this can be written equivalently as a $\Sigma^{1,b}_1$ 
formula and as a $\Pi^{1,b}_1$ formula,
\begin{align*}
&\exists W, \ \comp(e,t,W) \wedge \text{``$W$~is accepting'' \ and} \\
&\forall W, \ \comp(e,t,W) \rightarrow \text{``$W$~is accepting''},
\end{align*}
 since $V^0_2$ proves that
any two computations of the same machine are equal. 
Here ``$W$~is accepting'' means simply that the bit of $W$ coding the
ouput of the computation is 1.

Let $\phi$ be any formula and let $c \in \NN$. We define a sentence $\alpha^c_\phi$ expressing
that $\phi$ is in $\EXP / \poly$ with exponent~$c$, 
where we use the same $c$ to control the amount of advice
and the length of the computation:
\[
\alpha^c_\phi := 
\forall n {\in} \Log \, \exists e{<}2^{n^c} \,
	\forall x{<}2^n, \ \acc((e,x), 2^{n^c}) \leftrightarrow \phi(x).
\]
What this actually says is that $\phi$ is definable by the universal
deterministic machine, with a suitable time bound, using advice $e$ that depends
on the length of the input. Strictly speaking, being in $\EXP/\poly$  means
being accepted by \emph{some} exponential time machine with advice,
not just the universal machine, but $V^0_2$ is strong enough
to prove that the universal machine 
can simulate any other
machine using an appropriate code and time bound.\footnote{
Precisely, it proves that given a computation of a machine~$M$
we can construct a suitable computation of the universal machine
simulating~$M$.}

A language is in $\NEXP$ if and only if it is definable
by a $\St_1$ formula, and we will formalize $\NEXP$ by treating it as the
class of predicates defined by such formulas.
Again we could be more strict and insist that being in $\NEXP$
means being definable by the $\St_1$ formula 
``there is an accepting computation of~$M$ on this input''
 for some non-deterministic machine~$M$, but this would make no
difference to our argument. In particular, it is easy to construct an $M$
such that the formula above is equivalent  in $V^0_2$ to the formula ``$\exists Y \neg\PHP(x,Y)$'' 
which we use below.

We can now state precisely what we mean for a
theory $T$ to prove that $\NEXP \subseteq \EXP / \poly$:
we mean that for every $\St_1$ formula~$\phi$, there is
an exponent $c \in \NN$ such that $T \vdash \alpha^c_\phi$.
Thus to prove Proposition~\ref{pro:exp_poly_simple},
we need to show that
there is a $\St_1$-formula $\phi$ such that 
$V^0_2 + \EXPax + \{ \neg \alpha^c_\phi : c \in \NN \}$ is consistent.

Our proof uses a similar approach to~\cite{atserias2024consistency}. 
We first prove a lemma that gives us a model of $V^0_2 + \EXPax$
in which the pigeonhole principle fails at some size~$a$. Then we show that, if Proposition~\ref{pro:exp_poly_simple} were false, we would be able
to prove the pigeonhole principle in this model by induction.
Below, $\PHP(x,R)$ is the $\Sigma^b_\infty$ formula
expressing that $R$ is \emph{not} the graph of an injection from $x$ to~$x{-}1$.
We use the notation~$Z^e$ to mean ``the $e$th set coded by $Z$'' in the standard
way of coding many sets into one, that is, we use
 $x \in Z^e$ to mean~$(x,e) \in Z$.

\begin{lem} \label{lem:model}
For any $k \in \NN$ there is a model $M$,
with $n \in \mathrm{Log}$ and relations $R$ and $Z$ on~$M$ such that
\[
M \vDash
T_2(R, Z)
+ \forall e{<}{2^{n^k}} \comp(e, 2^{n^k}, Z^e) 
+ \neg \PHP(2^{n}, R).
\] 
\end{lem}

\begin{proof}
Suppose not. Then there is some $k \in \NN$ such
that
\[
T_2(R,Z) \vdash \forall n{\in}\mathrm{Log}, \,
\left[ \exists e{<}{2^{n^k}} \neg \comp(e, 2^{n^k}, Z^e) \right] 
\vee  \PHP(2^{n}, R).
\] 
Let us write $a$ for $2^n$. 
By the Paris-Wilkie translation of relativized bounded arithmetic into propositional logic
(\cite{paris1985counting} or see~\cite{krajicek1995bounded}),
for each $n$ there is a constant-depth Frege refutation~$\pi_a$, 
of size quasipolynomial in $a$, of the propositional formula
\[
\langle  \forall e{<}{2^{n^k}} \comp(e, 2^{n^k}, Z^e) \rangle
\ \wedge \
\langle \neg \PHP(a,R) \rangle
\]
where expressions in angled brackets represent translations of first-order
into propositional formulas, using $a$ as a size parameter.
The two conjuncts are in disjoint propositional variables,
standing for the bits of $Z$ on the left and the bits of $R$ on the right.

We now observe that there is an assignment $\alpha$ to the bits of $Z$ 
that satisfies the left-hand 
conjunct, which we can construct by actually running the  exponential time computations
on each input~$e$ and recording them as sequences of bits. 
Once we restrict $\pi_a$ by $\alpha$, what is left is a quasipolynomial-sized
constant-depth refutation of $\langle \neg \PHP(a, R) \rangle$, which is 
impossible by~\cite{krajivcek1995exponential, pitassi1993exponential}.
\end{proof}

We  derive Proposition~\ref{pro:exp_poly_simple} from the lemma.

\begin{proof}[Proof of Proposition~\ref{pro:exp_poly_simple}]
Let $\phi(x)$ be the formula $\exists Y \neg \PHP(x,Y)$.
We will show that $V^0_2 + \EXPax + \{ \neg \alpha^c_\phi : c \in \NN \}$ is consistent.
Suppose this is not the case.
Then there is some $c \in \NN$ such that $V^0_2 \vdash \EXPax \rightarrow \alpha^c_\phi$.

Writing $a$ for $2^n$, we have that $\alpha^c_\phi$ is $\forall n {\in} \mathrm{Log} \, \Phi(a)$
where  
\[
\Phi(a) \ := \
\exists e{<}2^{n^c} \,
	\forall x{<}2^n, \ \acc((e,x), 2^{n^c}) \leftrightarrow \phi(x)
\]
expresses that $\phi$ is
in $\EXP / \poly$ at length~$n$. We can move the quantifier $\forall n {\in} \mathrm{Log}$
outside the implication $\EXPax \rightarrow \alpha^c_\phi$
and apply Parikh's theorem~(\cite{parikh1971existence} or see~\cite{buss1985bounded}) 
to obtain that for some~$k$, which we may assume is larger than $c$,
\begin{equation}\label{eq:Phi}
V^0_2 \vdash \forall n {\in}\mathrm{Log}, \,
\left[ \forall e,t{<}2^{n^k} \exists W \comp(e,t,W) \right] \longrightarrow \Phi(a).
\end{equation}

Now let $M, n, R, Z$ be as given by Lemma~\ref{lem:model},
so that $M$ satisfies simultaneously
$T_2(R, Z)$,
$\forall e{<}{2^{n^k}} \comp(e, 2^{n^k}, Z^e)$ 
and $\neg \PHP(2^{n}, R)$. 
Since $\comp$ and $\neg \PHP$ do not contain any universal set-sort
quantifiers, we may assume without loss of generality that  $M$
is a model of $V^0_2$, since we can make it into such a model by adding to it every bounded
set definable in $M$ by a $\Sigma^b_\infty(Z,R)$ formula
with number-sort parameters.
Then in particular $M \vDash \forall e,t{<}2^{n^k} \exists W \comp(e,t,W)$,
since every such~$W$ is encoded inside $Z^e$, so $M \vDash \Phi(a)$ by~(\ref{eq:Phi}).

Thus we have $M \vDash V^0_2 +  \neg \PHP(a, R)$
and simultaneously, expanding $\Phi(a)$, that for some~$d \in M$,
\[
M \vDash \forall x{<}a, \ \acc((d,x), 2^{n^c}) \leftrightarrow \exists Y \neg \PHP(x, Y).
\]
Since $Z$ uniformly contains all computations running in time up to $2^{n^k}$, 
we can replace $\acc((d,x), 2^{n^c})$ with an equivalent 
$\Sigma^b_\infty(Z)$ formula $\theta(Z, d, x)$
where $\theta$  simply looks up in $Z$ the final state of the computation on input $(d,x)$.
Since induction holds for $\Sigma^b_\infty(Z)$ formulas in $M$, 
and $M \vDash \forall Y \, \PHP(x, Y)$
for every standard $x$, we conclude by induction in~$M$ that there
is some $x<a$ in $M$ with $\forall Y \, \PHP(x, Y)$ but $\exists Y' \neg \PHP(x+1, Y')$.
But in $V^0_2$, given a relation $Y'$ failing $\PHP$ at $x+1$ we can, by
changing at most two pigeons, construct
a relation $Y$ failing $\PHP$ at $x$, as in~\cite{atserias2024consistency}. So this is a contradiction.
\end{proof}

We  go on to show the full Theorem~\ref{the:exp_poly},
which replaces $V^0_2 + \EXPax$ with the stronger theory $V^0_2 + \forall \St_1(\NN)$.
Let us observe in passing that  
 $V^0_2 + \EXPax$ already captures a nontrivial amount of the strength of $V^1_2$, 
 namely
it proves all the $\Pi_1$ number-sort consequences of~$V^1_2$~\cite{kol2011so}
and in fact every $\forall\St_1$ consequence of~$V^1_2$~\cite{beckmann2014improved}.

\begin{proof}[Proof of Theorem~\ref{the:exp_poly}]
We  want to show that $V^0_2 + \forall \St_1(\NN)+ \{ \neg \alpha^c_\phi : c \in \NN \}$ is consistent,
where $\phi(a)$ is again the formula $\exists Y \neg \PHP(a,Y)$.
Suppose not. Then by compactness, and using the fact that in $V^0_2$ we can
combine any finite number of $\forall \St_1$ formulas into one formula,
 there is $c \in \NN$ and a single $\Sigma^b_\infty(U)$ formula $\theta(e, U)$ such that
  \[
 V^0_2 \vdash \EXPax \wedge \forall e \exists U \theta(e,U) \rightarrow \alpha^c_\phi
 \]
and $\forall e \exists U \theta(e,U) $ is true in the standard model.
We include $\EXPax$ here to guarantee natural behaviour
of the relation $\acc$ which occurs in $\alpha^c_\phi$.

We now imitate the proof of Proposition~\ref{pro:exp_poly_simple}.
As in that proof, we can use Parikh's theorem to bound, by some term in $a=2^n$,
 the values of $e$ for which we need to witness $\theta$. That is, we have for some $k>c$ that
\begin{equation*}
V^0_2 \vdash \forall n {\in}\mathrm{Log}, \,
\left[
\forall e,t{<}2^{n^k} \exists W \comp(e,t,W) 
\wedge \forall e{<}2^{n^k} \exists U \theta(e,U) \right]
\longrightarrow \Phi(a).
\end{equation*}
Finally we strengthen Lemma~\ref{lem:model} 
to give a model~$M \vDash T_2(R,Z,U)$ which, in addition to the conditions
in Lemma~\ref{lem:model} , also satisfies
$\forall e{<}2^{n^k} \theta(e, U^e)$.
 We can do this using the same argument that we used before for the axiom
 $\EXPax$. Namely, we
can always find an assignment to the variables for $U$ which satisfies the propositional translation
$\langle \forall e{<}2^{n^k} \theta(e, U^e) \rangle$, because the
sentence $\forall e \exists U \theta(e,U)$ is true in the standard model.
\end{proof}

\section{$\NEXP$ and $\coNEXP$}

We prove Theorem~\ref{the:conexp}, that 
$V^0_2 + \forall \St_1(\NN) \not \vdash \NEXP  \subseteq \coNEXP$. 
The argument is slightly different as it does not rely so directly on induction, and
it is not clear if a similar argument can work for $\coNEXP / \poly$.

\begin{proof}[Proof of Theorem~\ref{the:conexp}]
We will again take  the formula $\exists Y \neg \PHP(x, Y)$ as
our relation in $\NEXP$. We will show that the theory does not prove
it is in $\coNEXP$. Suppose for a contradiction that it does, 
by which we mean that there is a  $\Sigma^b_\infty(S)$ formula $\chi(x,S)$,
with no other free variables, such that 
\begin{equation}\label{eq:NEXP_coNEXP}
V^0_2 + \forall \St_1(\NN) \vdash \forall x, \ \exists Y \, \neg \PHP(x,Y) \leftrightarrow \forall S \, \chi(x,S)
\end{equation}
(but see the discussion at the end of this section for a different formalization,
which allows a parameter in~$\chi$ and
for which our argument does not work).

As in the proof of Theorem~\ref{the:exp_poly}, we can use compactness  
to replace $\forall \St_1(\NN)$ with a single true $\forall \St_1$ sentence
$\forall e \exists U \theta(e,U)$,  then move this to the right-hand side, inside
the scope of~$\forall x$, and use
Parikh's theorem to bound~$e$. We also only take the left-to-right direction
of the equivalence in~(\ref{eq:NEXP_coNEXP}). We obtain in this way, for some $k\in\NN$, that
\begin{equation*}
V^0_2 \vdash \forall x, \left[
\forall e {<} 2^{|x|^k} \exists U \theta(e,U) \
\wedge   \ \exists Y  \neg \PHP(x,Y) \right] \rightarrow \forall S \, \chi(x,S).
\end{equation*}
Hence to get a contradiction it is enough
to find a model $M$ with an element $a$ 
and relations $R$, $S$, $U$ on $M$ such that
\[
M \vDash
T_2(R,S,U)
+ \forall e{<}2^{|a|^k} \theta(e, U^e) 
+ \neg \PHP(a,R) + \neg \chi(a,S).
\]

\sloppypar
Suppose there is no such $M$. Then as in Lemma~\ref{lem:model} there are
quasipolynomial-size refutations $\pi_a$ of the propositional translation
\[
\langle \forall e{<}2^{|a|^k} \theta(e, U^e) \rangle
\wedge \langle \neg \PHP(a, R) \rangle
\wedge \langle \neg \chi(a,S) \rangle.
\]
We can construct an assignment to the
$U$ variables that satisfies the first conjunct $\langle \forall e{<}2^{|a|^k} \theta(e, U^e) \rangle$
exactly as in the proof of Theorem~\ref{the:exp_poly}.
For the last conjunct $ \langle \neg \chi(a,S) \rangle$, observe that 
in the standard model $\exists Y \, \neg \PHP(x,Y)$ is false for all $x$. Therefore, 
by~(\ref{eq:NEXP_coNEXP}) and the fact that 
our theory $V^0_2 + \forall \St_1(\NN)$
 is sound,
we have that $\forall x \exists S  \chi(x,S)$ holds in the standard model,
so there is an assignment to the $S$ variables satisfying~$ \langle \neg \chi(a,S) \rangle$.
Applying these two partial assignments to $\pi_a$ we get quasipolynomial-size
refutations of  $\langle \neg \PHP(a, R) \rangle$, which is impossible.
\end{proof}

We  discuss briefly what happens if you try to extend this argument
to show unprovability of
$\NEXP \subseteq \coNEXP / \poly$.
A natural formalization of this inclusion is: for every
$\Sigma^b_\infty(Y)$ formula $\phi(x, Y)$,
there is a $\Sigma^b_\infty(S)$ formula $\chi(e,x,S)$ such that
\begin{equation} \label{eq:coNEXP/poly}
\forall n \in \mathrm{Log} \,
\exists e  \, \forall x {<} 2^n, \,
\exists Y \phi(x,Y) \leftrightarrow \forall S \chi(e,x,S).
\end{equation}
Suppose this is provable, and let us ignore for simplicity the dependence of~$e$ on~$n$.
We can try to imitate the proof of Theorem~\ref{the:conexp}. We put
$\phi(x,Y) := \neg \PHP(x,Y)$, so that $\exists Y \phi(x,Y)$ is always false in the standard
model, and use the right-to-left direction of~(\ref{eq:coNEXP/poly}) to get some $e_0 \in \NN$
for which $\forall x \exists S \neg \chi(e_0, x, S)$ is true in the standard model.
We can then use this to construct a model $M$ with relations $R, S$ satisfying, among
other things, that $\phi(a,R)$ and $\neg  \chi(e_0, a, S)$.
We would like $M$ to in some way falsify the left-to-right direction of~(\ref{eq:coNEXP/poly}),
but it does not, since that asserts that for some $e$ we have 
$\exists Y \phi(a,Y) \rightarrow \forall S \chi(e,a,S)$, and we only have that this fails for~$e_0$,
which may be different from~$e$.

Finally let us mention a more liberal formalization of $\coNEXP$, for which
we do not know how to prove Theorem~\ref{the:conexp}. Our formulation of 
$T \vdash \NEXP \subseteq \coNEXP$ in Theorem~\ref{the:conexp} is: 
for every formula $\phi(x)$ representing a $\NEXP$ predicate, there is a 
$\Sigma^b_\infty(S)$ formula $\chi(x,S)$ such that 
$T \vdash \forall x, \ \phi(x) \leftrightarrow \forall S \chi(x,S)$.
A weaker version of this allows $\chi$ to take an extra parameter~$e$, and has the form:
for every formula $\phi(x)$ representing a $\NEXP$ predicate, there is a 
$\Sigma^b_\infty(S)$ formula $\chi(e, x,S)$ such that 
$T \vdash \exists e \forall x, \ \phi(x) \leftrightarrow \forall S \chi(e, x,S)$.
This is similar to  our formalization of $T \vdash \NEXP \subseteq \coNEXP / \poly$,
except that $e$ no longer varies with the binary length of~$x$.
We do not know how to prove Theorem~\ref{the:conexp} for this version, 
for the reason sketched in the previous paragraph.

It is helpful to think of $\chi(e,x,S)$ here as expressing acceptance by a universal $\coNEXP$ machine
which takes a description of a machine~$e$. In Theorem~\ref{the:conexp} we ask that there is some fixed,
standard~$e \in \NN$ for which $\phi(x)$ is equivalent to $\forall x \chi(e,x,S)$;
in the weaker version the machine~$e$ may be different in different models of~$T$, and may be 
nonstandard. We can call these different formalizations respectively 
$\coNEXP$ \emph{without parameters} and $\coNEXP$ \emph{with parameters}.
For example Kraj\'{i}\v{c}ek~\cite[Section~7.6]{krajicek1995bounded} takes the version of 
a complexity class with parameters as the more basic one;
in the notation of~\cite{krajicek1995bounded} the class without parameters
in Theorem~\ref{the:conexp} would be denoted rather $\coNEXP^{\mathrm{stan}}$.

\section*{Acknowledgement}
I am grateful to Albert Atserias, Sam Buss and Moritz M\"{u}ller for 
helpful discussions on this topic and for comments on earlier versions of this work,
and to an anonymous referee for raising the issue of different formalizations 
of~$\coNEXP$.

\bibliographystyle{alphaurl}
\bibliography{main}

\end{document}